\documentclass[journal]{IEEEtran}

\usepackage{cite}

\ifCLASSINFOpdf
\else
\fi

\usepackage[cmex10]{amsmath}

\usepackage{amsthm}
\usepackage{amsfonts}
\usepackage{mathtools}
\usepackage{graphicx}
\newtheorem{definition}{Definition}
\newtheorem{theorem}{Theorem}[section]
\newtheorem{lemma}[theorem]{Lemma}

\newtheorem{corollary}[theorem]{Corollary}
\newtheorem{example}{Example}
\usepackage{amssymb}
\usepackage{mathtools}

\DeclarePairedDelimiter\floor{\lfloor}{\rfloor}
\begin{document}
%
\title{Combining Conventional Cryptography with Information Theoretic Security}

\author{Jason~Castiglione
\thanks{J. Castiglione is with the Department
of Electrical and Engineering, University of Hawaii Manoa,
Honolulu, HI, 96816 USA e-mail: jcastig@hawaii.edu }
}

\maketitle

\begin{abstract}
This paper highlights security issues that can arise when incorrect assumptions are made on the capabilities of an eavesdropper. In particular, we analyze a channel model based on a split Binary Symmetric Channel (BSC). Corresponding security parameters are chosen based on this channel model, and assumptions on the eavesdroppers capabilities. A gradual relaxation of the restrictions on the eavesdropper's capabilities will be made, and the resulting loss of security will be quantified. An alternative will then be presented that is based on stochastic encoding and creating artificially noisy channels through the usage of private keys. The artificial channel will be constructed through a deterministic process that will be computationally intractable to reverse. 
\end{abstract}

\section{Introduction}

\IEEEPARstart{A}{ path} to securing communications through the use of information theoretic notions was started by Shannon in the 1940's. In his papers he related fundamental notions, such as entropy, to the secrecy of cryptographic systems \cite{wtap:shan}. Inspired by the work of Shannon, Wyner published a paper proving that it was theoretically possible to secure communications solely through the choice of an encoding scheme for a specific channel model \cite{wtap:wyner}. Wyner's security model capitalizes on the eavesdropper receiving a noisier copy of what the intended user receives. \par

In general, for communications security, assumptions on the capabilities of eavesdroppers are required to design systems. In Wyner's paper, the presumption is on the channel noise of the eavesdropper. On the other hand, in contemporary cryptography it is assumed that the eavesdropper is subject to certain computational limitations. It is accepted that many people have attempted to break a cryptosystem, and their best efforts to recover plaintext require an infeasible amount of time. In particular, the computational complexity of factoring products of primes is unknown, although it is widely accepted to be a hard problem. Thus the only guarantee of the security of these systems is that numerous people have attempted to attack the system and have not succeeded.  \par

The question for physical layer security is what circumstances will allow similar conclusions about the eavesdropper's capabilities as in cryptography. In physical layer security, channel noise combined with a special encoder provide the security. Thus for Wyner's method, how can one ensure that the eavesdropper's channel has specific noise characteristics, thus limiting his/her decoding capabilities? \par
The goal in this paper is to work with a well-studied channel model, and gradually relax restrictions on the eavesdropper's capabilities. As these relaxations are made, the effect on the security of the system will be evaluated. We will show that if the real channel characteristics deviate even slightly from the assumptions, all guarantees of security are lost.  \par

Since in a practical system, no assumptions (either on the intended user's or the eavesdropper's channel) are verifiable, we can not claim $100\%$ security. The intended users rarely have control over the eavesdropper's access (especially in the wireless setting) or the eavesdropper's behavior, so basing a security system on questionable assumptions is risky. On the other hand, Wyner's basic idea can be easily combined with cryptographic methods to achieve a secure system based on computationally intractable problems. For example, we can artificially create a channel that is information-theoretically secure under the assumption that the shared key can not be recovered, i.e., the recovery of the key is computationally intractable. \par

The paper is organized as follows. In Section II, the channel model is presented. The general results for wiretap coding and achieving secrecy are discussed in Section III. Section IV quantifies the amount of secrecy lost when the eavesdropper has better access than anticipated. In Section V, a shared key cryptosystem is presented from Mihaljevic\cite{wtap:mishalpn}, that sets the path for a more general method of combining stochastic encoding with computationally intractable problems. \par

The notation for this paper is as follows. Random variables will be represented by uppercase letters ($S,X,Y,Z$), and their instances by lower case letters ($s,x,y,z$). Caligraphic uppercase letters will represent the domain of the corresponding random variable ($\mathcal{S},\mathcal{X},\mathcal{Y},\mathcal{Z}$). Since we will consider a sequence of codes, $C_M$, we will  use a $M$ subscript to denote the corresponding parameters, i.e., $C_M$ will have rate $\frac{K_M}{N_M}$.   \par 

\section{CHANNEL MODEL}
The communication model used for this paper is a broadcast channel with confidential messages (BCC) \cite{wtap:krner}, and depicted in Fig. \ref{fmod}. In general, the BCC model has confidential messages and public messages. The BCC model is a generalization of the Wyner wiretap model. In this paper, there will only be confidential messages. The model in this paper consists of a memoryless concatenated additive white Gaussian noise (AWGN) channel followed by a 2-level and L-level analog to digital (A/D) converters on the main and wiretap channels, respectively.\par
\begin{figure}[!htbp]
\centering
\includegraphics[width=3in]{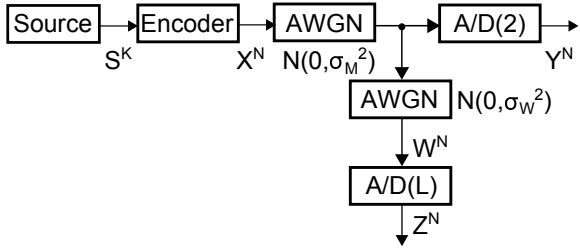}
\caption{Channel Model}
\label{fmod}
\end{figure}
\begin{definition}
 (Stochastic encoder) A \textbf{\textit{stochastic encoder}}, $f_E$, with rate $\frac{K}{N}$, input $S^K$, and output $X^N$, is a channel with transition probability $f(x^N|s^K)$, that satisfies;
 \begin{enumerate}
 \item $\sum_{x^N} f(x^N|s^K)=1, \text{ for any } s^K \in \mathcal{S}^K$
 \item If $f( x^N|s^K) > 0$ and $t^K \neq s^K$, then $f( x^N|t^K)=0$
 \end{enumerate}
\end{definition}
The input to the stochastic encoder is a uniform random variable, $S^K \in \{1,2,...,2^K\}$, which is the source message, and has entropy $$H(S^K)=K \cdot H(S)= K.$$
The stochastic encoder, with rate $\frac{K}{N}$, outputs $X^N \in \mathcal{X}^N$. For our channel model, we let $\mathcal{X}^N = \{ {-} 1 ,1 \}^N $. In this model there is a legitimate receiver and an eavesdropper. The legitimate user receives $Y^N$. When $L=2$, the eavesdropper receives $Z^N$. To distinguish the case when $L>2$, we will say that the eavesdropper receives $\hat{Z}^N$. There is independent AWGN noise in both the main and wiretap channel. We assume that the eavesdropper has a noise variance greater than that of the legitimate receiver. When just considering the input and output, the channels are both discrete memoryless channels (DMC).\par
 
The legitimate channel users will construct a security system based on the assumption that the eavesdropper is restricted to a two-level A/D converter, and $\sigma_W^2>0$. The goal in this paper is to quantify the effect of making an incorrect assumption about $L$, and suggest an improvement through the use of a shared key based cryptosystem.\par
In the case $L=2$, it follows that the main and eavesdropper channels can be individually modeled as binary symmetric channels, $BSC(p)$, and $BSC(p_W)$ respectively. The corresponding crossover probabilities are  $p= \Phi(-1/\sqrt{\sigma_M^2})$, and $p_W= \Phi(-1/\sqrt{\sigma_M^2+\sigma_W^2})$ where $\Phi$ is the CDF for the normal distribution. Since $\sigma_W^2>0$, it follows that $0 \leq p < p_W<1/2$. This means that the capacity of the wiretap channel is lower than the capacity of the main channel.\par
\begin{definition}
(Stochastically degraded channel) The channel, $(X,p_{Z|X},Z)$  , is said to be \textbf{\textit{stochastically degraded}} with respect to the channel, $(X,p_{Y|X},Y)$, if there exists a channel, $(Y,p_{Z|Y},Z)$, such that for every $(x,z) \in \mathcal{X}  \times \mathcal{Z},$
\begin{equation}
p_{Z|X}(z|x)=\sum_{y \in \mathcal{Y}} p_{Z|Y}(z|y)p_{Y|X}(y|x)
\end{equation}

\end{definition}

First note, the concatenation of $BSC(p)$  and $BSC(p_Y )$ channels is a $BSC(p_W)$ channel, where  $p_W = p(1-p_Y )+(1-p)p_Y$. Thus with regards to conditional probabilities, the concatenated and split BSCs are equivalent, see Fig. \ref{stocheq}. It follows that the wiretap channel is stochastically degraded with respect to the main channel \cite{wtap:onbsc}. This is important since the results in Wyner's wiretap paper pertain strictly to concatenated channels, yet they are extendable to split channels under these conditions.

\begin{figure}[!htbp]
\centering
\includegraphics[width=3 in]{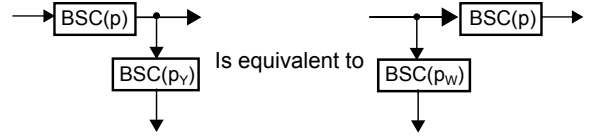}
\caption{Equivalence with regards to conditional probability}
\label{stocheq}
\end{figure}

\section{Secrecy Capacity}

It is clear that we have a channel model where the eavesdropper is “worse off” than the intended receiver. Now the question is how the additional uncertainty in the eavesdropper's channel leads to secrecy. To quantify the ability to send confidential messages over a broadcast channel, we define first a measure of uncertainty that the eavesdropper has with regards to the original message.\par
\begin{definition} (Equivocation) The \textbf{\textit{equivocation}} of the eavesdropper is defined as 	
\begin{equation}
\Delta \triangleq \frac{H(S^K|Z^N)}{K}
\end{equation}

\end{definition} \par
In the event  $H(S^K | Z^N )=H(S^K)=K$, i.e., $Z^N$ provides no information about the original source, the eavesdropper's only available method is to guess according to the distribution of the source message. In designing a code, the objective is to achieve security through maximizing the equivocation. That is if $H(S^K|Z^N)=H(S^K)$, then the system is secure. 
\begin{definition} (Wiretap Code) A \textbf{\textit{wiretap $(2^K,N)$ code}}, $C_N$, for a BCC consists of
\begin{itemize}
	\item Message set, $\mathcal{S}^K = \{1,2, \dots , 2^K \} $
	\item Stochastic encoder, $f_E:\mathcal{S}^K \rightarrow \mathcal{X}^N$
	\item Decoder, $g:\mathcal{Y}^N \rightarrow \mathcal{S}^K$
\end{itemize}
\end{definition} 

Secrecy in these systems is introduced through the use of the stochastic encoder. The encoder will map one message to different outputs such that when the eavesdropper has enough error, any attempt to decode will result in a random message.

\begin{example}
Assume that the main channel is noiseless, and the wiretap channel is a $BSC(0.25)$. Let $\mathcal{S}^K =\{1,2\},H(S^K )=1$. Define the stochastic encoder as follows,
\begin{itemize}
\item[]  $1 \mapsto (-1,-1)$ or $(1,1)$ with equal probability 
\item[]  $2 \mapsto (1,-1)$ or $(-1,1)$ with equal probability
\end{itemize}
We can calculate the equivocation of this simple coding scheme, which results in 
\begin{align}
H(S|Z^2)=&- [ \left( (1-p_W)^2 +p_W^2 \right) * \log \left( (1-p_W)^2 +p_W^2 \right) \\
&+ \left( 2(1-p_W)p_W \right) * \log \left( 2(1-p_W)p_W \right) ] \\
\approx & 0.954 
\end{align}
Observe in the case we transmit the message without any encoding over the channel, the eavesdropper will have equivocation $h(p_W) \approx 0.81$. Thus we have increased the equivocation of the eavesdropper using a wiretap $(2,2)$ code. Increasing equivocation through a random one-to-many mapping is the notion behind using a stochastic encoder for security. 
\end{example}
\subsection{Achievability}
In the above example, we presented a rate $\frac{1}{2}$ stochastic encoder that increased the eavesdropper's equivocation. The question remains as to what are the limits with regards to possible rates, and equivocation. First we define requirements for a rate-equivocation pair to be achievable.
 
\begin{definition}
(Achievability)
A rate-equivocation pair, $(R,d)$, is \textbf{achievable} if there exists a series of wiretap $(2^{K_M },N_M )$ codes, $ \{ C_M \}_{M\geq 1}$ such that
\begin{enumerate}
\item  $\lim_{M \rightarrow \infty} \frac{H(S) \cdot  K_M}{N_M}=R$ (\textbf{Rate})
\item $\lim_{M \rightarrow \infty} P_e (C_M)=0$ (\textbf{Probability of Error})
\item  $\lim_{M \rightarrow \infty} \frac{H(S^{K_M}|Z^{N_M})}{K_M}\geq d $ (\textbf{Equivocation})
\end{enumerate}

\end{definition} \par
The second achievability condition assures that the legitimate user will be able to decode the confidential message. The third achievability condition is called weak secrecy, and puts a lower bound on the wiretapper's equivocation. Now we have the framework to define the secrecy capacity.\par

\begin{definition}
(Secrecy Capacity)
The secrecy capacity of a BCC with no common message is
\begin{equation}
C_S \triangleq  \sup_R \{R:(R,H(S)) \:  is \: \text{achievable} \} 
\end{equation}

\end{definition}
Thus $C_S$ is the maximum rate at which error free communication is possible over the main channel, with maximum equivocation over the wiretap channel. Furthermore the region of achievable rates is characterized by the following generalization of the Wyner's results.
\begin{theorem}
(Region of Achievability)(Korner and Csiszar \cite{wtap:krner})
Given a BCC with no common message the region, 
$$\mathfrak{R}_{1e}\triangleq \left\{ (R,\frac{R \cdot d}{H(S)} ) | (R,d) \text{ is achievable} \right\} $$ is a closed convex set for which there exist random variables, $U$ and $V$, such that  $U \rightarrow V \rightarrow X \rightarrow Y$ and $U \rightarrow V \rightarrow X \rightarrow Z$ are Markov chains.  Furthermore, the conditional distribution of $Y$ (respectively $Z$) given $X$ characterizes the main (respectively wiretap) channel and
\begin{enumerate}
\item $0 \leq d \leq H(S)$
\item $R \cdot d \leq H(S) \cdot  \left[ I(V;Y|U) - I(V;Z|U) \right]$
\item $R \leq I(V;Y|U) + \min \left[ I(U;Y), I(U;Z) \right]$
\end{enumerate} 
In particular,
\begin{equation}
 C_S= \max_{ {V \rightarrow X \rightarrow Y }\choose{V \rightarrow X \rightarrow Z}} \left[ I(V;Y) -I(V;Z) \right]
\end{equation}

\end{theorem}

 Since the wiretap channel is stochastically degraded with respect to the main channel, \cite{wtap:krner} and \cite{wtap:bloch} give us that for our channel model, the following corollary holds.
\begin{corollary}
\label{acst}
(Achievability for Stochastically Degraded Channels)(Korner and Csiszar \cite{wtap:krner}) If the wiretap channel is stochastically degraded with respect to the main channel, then the above region, $\mathfrak{R}_{1e}$, simplifies to those pairs $(R,\frac{R \cdot d}{H(S)} )$ such that

\begin{enumerate}
\item $0 \leq R \cdot d \leq H(S) \cdot  \left[ I(X;Y)-I(X;Z) \right] $
\item $0 \leq d \leq H(S) $
\item $0 \leq R \leq I(X;Y) $
\end{enumerate}
In particular, 
\begin{equation}
C_S= \max_{ p_X} \left[ I(X;Y) -I(X;Z) \right]
\end{equation}  
where the maximization is over probability distributions, $p_X$, on $X$.
\end{corollary} 
 
Thus for our channel model, in the case $L=2$, it follows, 
\begin{align}
C_S &=  \max_{p_X } [I(X;Y)-I(X;Z)] \\
&= \max_{p_X} [ H(Y) - H(Y|X) - H(Z) + H(Z|X) ]  \\
&= \max_{p_X} [ \mathcal{H}(p_W)-\mathcal{H}(p) + H(Y)-H(Z) ] \\
&\leq \mathcal{H}(p_W)-\mathcal{H}(p) \label{eq:bla}
\end{align}
Observe Eq. (\ref{eq:bla}) follows because a BSC does not decrease entropy, i.e., $H(Y) \leq H(Z)$. Furthermore note for $X$ uniform, $H(Y)=H(Z)$, thus $$C_S=\mathcal{H}(p_W)-\mathcal{H}(p). $$

This means that if the legitimate users believe that the eavesdropper has a A/D(2) converter, they will also believe that they can securely transmit over the main channel at $h(p_W) - h(p)$ bit per channel use. Thereby under this belief, the eavesdropper will not be able to do better than random guessing of the source message. 

\subsection{Code Parameters}
From this point, since $(C_S,1)$ is achievable, we will construct a sequence of random linear codes as specified in \cite{wtap:onbsc}.  Let $\{ \epsilon_M \}_{M \geq 1}$ be a sequence of positive numbers $\epsilon_1, \epsilon_2,\dots$, such that $\epsilon_M < 1/M$. Then for each $\epsilon_M$, we can choose $N_M$  large enough such that there exists a wiretap code, $C_M$, satisfying;

\begin{enumerate}
\item $\frac{K_M}{N_M} \geq h(p_W)-h(p)-\epsilon_M$
\item $\frac{H(S^{K_M}|Z^{N_M})}{K_M} \geq 1-\epsilon_M$
\item $P_e(C_M) \leq \epsilon_M$
\end{enumerate} \par

 Note, each $C_M$ is specified by a binary $(N_M-K_{M,2} ) \times (N_M )$ matrix, $H_{M,2}$, with the following parameters;

 \begin{align}
K_{M,2} &= \floor{N_M [1-h(p)-2 \epsilon_M ]} \\
K_{M,1} &=\floor{N_M [1-h(p_W )-2\epsilon_M]} \\  
K_M     &=K_{M,1}-K_{M,2} 
 \end{align} 
Then to encode a source message, $s^{K_M}$, we randomly choose a solution $x^{N_M}$ of the equation,
\begin{equation}
x^{N_M } H_{M,2}^T=[\textbf{0} \; s^{K_M}  ]
\end{equation}
as the encoded message. To decode, a typical set decoder is used. In practice, this is not a realistic decoding scheme, but will suffice in this paper to allow us to calculate effects on equivocation. Han Vinck and Chen \cite{wtap:onbsc} have an elegant proof that this coding method will suffice and achieves $(C_S,1)$.

\section{Soft Decoding}
The rate-equivocation pair $(C_S,1)$ is achievable, and we have an associated sequence of wiretap codes, $\{C_M\}_{M \geq 1}$. We will now relax the restriction on the A/D(L) converter to analyze the effect on equivocation. Our goal will be to look at the effects of increasing $L$ on the equivocation. \par

The \textit{equivocation loss} is defined as the difference between the believed equivocation, and the actual equivocation. In other words, the equivocation loss is the amount of information (in bits per channel use) that is leaked to the eavesdropper as a result of a false belief (wrong assumption) by the legitimate users. In particular, we are interested in the asymptotic value of the equivocation loss,
\begin{equation}
\frac{H(S^K|Z^N)-H(S^K|\hat{Z}^N)}{K}
\end{equation}
as $L \rightarrow \infty$, where $\hat{Z}$ is dependent on L. Since we are unaware of the capabilities of the eavesdropper, our assumption will be that arbitrarily level A/D converters are available.

\begin{theorem}
(Equivocation Loss) \\
$$\lim_{M \rightarrow \infty} \frac{H(S^{K_M}| Z^{N_M})- H(S^{K_M}| \hat{Z}^{N_M})}{K_M} =$$ $$ \frac{h(p_W) -1 + I(X;\hat{Z})}{h(p_W)-h(p)}$$
\end{theorem}

\begin{proof}
This proof uses the analysis of the BSC-BSC wiretap channel given by Han Vinck, and Chen. First note by \cite{wtap:onbsc}, we have the following equalities,
\begin{align}
H(S^{K_M}|&Z^{N_M})=H(S^{K_M} , Z^{N_M})- H(Z^{N_M}) \\
=& H(S^{K_M} , X^{N_M}, Z^{N_M}) \nonumber \\ 
&- H(X^{N_M} | S^{K_M} , Z^{N_M}) - H(Z^{N_M}) \\
=& H(S^{K_M} , X^{N_M}| Z^{N_M}) \nonumber \\
&- H(X^{N_M} | S^{K_M} , Z^{N_M})\\
=& H(X^{N_M}| Z^{N_M}) \nonumber \\ 
&- H(X^{N_M} | S^{K_M} , Z^{N_M})\\
=& H(X^{N_M}| Z^{N_M})- H(X^{N_M} | Y^{N_M})  \nonumber \\ 
& + H(X^{N_M} | Y^{N_M})- H(X^{N_M} | S^{K_M} , Z^{N_M})\\
=& I(X^{N_M} ; Y^{N_M})- I(X^{N_M} | Z^{N_M})  \nonumber \\ 
& + H(X^{N_M} | Y^{N_M})- H(X^{N_M} | S^{K_M} , Z^{N_M})\\
=& N_M [I(X ; Y)- I(X | Z) ]  \nonumber \\ 
& + H(X^{N_M} | Y^{N_M})- H(X^{N_M} | S^{K_M} , Z^{N_M})
\end{align}
The same holds for $\hat{Z}^{N_M}$ , replacing $ Z^{N_M }$ as needed. Next we substitute both equations into the difference of equivocation, resulting in,
$$ \frac{H(S^{K_M}|Z^{N_M})-H(S^{K_M}|\hat{Z}^{N_M})}{K_M} = \frac{N_M [I(X ; \hat{Z})- I(X | Z) ] }{K_M} $$ 
$$+  \frac{[H(X^{N_M} | S^{K_M} , \hat{Z}^{N_M})-H(X^{N_M} | S^{K_M} , Z^{N_M})] }{K_M}$$
Since conditioning reduces entropy, and $ X^{N_M } \rightarrow \hat{Z}^{N_M } \rightarrow Z^{N_M }$  is a Markov chain, it follows,
\begin{align}
H(X^{N_M } |S^{K_M },\hat{Z}^{N_M } )&=H(X^{N_M } |S^{K_M },Z^{N_M },\hat{Z}^{N_M} ) ) \\
 & \leq H(X^{N_M} |S^{K_M},Z^{N_M} )
\end{align}
From \cite{wtap:onbsc}, we have,
\begin{enumerate}
\item $	H(X^{N_M} |S^{K_M},Z^{N_M} ) \leq h(P_{ew} )+P_{ew}*K_{M,2} $
\item For any $ 0 < \lambda < \frac{1}{2}$ , we can choose $N_M$ large enough so $P_{ew} \leq \lambda$
\item $ \lambda$ is dependent on $N_M$, and $\lim_{M \rightarrow \infty} \lambda =0 $
\end{enumerate}
Thus we can bound the magnitude of the difference of the conditional entropies and take the limit,
\begin{align}
0 \leq& \lim_{M \rightarrow \infty}  \frac{[H(X^{N_M} | S^{K_M} , Z^{N_M})-H(X^{N_M} | S^{K_M} , \hat{Z}^{N_M})] }{K_M} \\
\leq& \lim_{M \rightarrow \infty} \frac{2(h(P_{ew})+ P_{ew}K_{M,2})}{K_M} \\
\leq& \lim_{M \rightarrow \infty} \frac{2(h(\lambda)+ \lambda K_{M,2})}{K_M} \\
\leq& \lim_{M \rightarrow \infty} \frac{2(h(\lambda)+ \lambda (N_M[1-h(p)-2\epsilon_M]+1))}{N_M(h(p_W)-h(p)-\epsilon_M)} \\
=& \lim_{M \rightarrow \infty} \left[ \frac{2h(\lambda)}{N_M(h(p_W) -h(p) -\epsilon_M )} \right] \nonumber \\
 + & \lim_{M \rightarrow \infty} \left[ \frac{2\lambda(N_M[1-h(p)-2\epsilon_M]+1)}{N_M(h(p_W) - h(p) -\epsilon_M )} \right]  \\
=& 0
\end{align} \par
The last equality follows since for the first term in the sum, the numerator goes to $0$, and the denominator grows. The second term has $\lambda \rightarrow 0$ , while the remainder converges to a constant. To complete the proof since,
$$
h(p_W) - h(p) \geq \frac{K_M}{N_M} \geq h(p_W) - h(p) -\epsilon_M
$$

Then,
$$
\lim_{M \rightarrow \infty} \frac{K_M}{N_M} =h(p_W) - h(p)
$$
Hence,
\begin{align}
\lim_{M \rightarrow \infty} & \frac{H(S^{K_M}| Z^{N_M})- H(S^{K_M}| \hat{Z}^{N_M})}{K_M} \\
&= \lim_{M \rightarrow \infty} \frac{N_M[I(X;\hat{Z})-I(X;Z)]}{K_M}\\
&= \frac{h(p_W) -1 + I(X;\hat{Z})}{h(p_W)-h(p)}
\end{align}
\end{proof}
The equivocation loss theorem allows us to calculate the loss in uncertainty of the eavesdropper due to using an A/D(L) converter, when it is assumed by the legitimate users that the eavesdropper is using an A/D(2) converter.
\begin{lemma}
(Approximation Lemma) \\
For any $ \epsilon > 0$, there exists a $L>0 $ and A/D(L) converter such that
\begin{equation}
I(X;W)-I(X;\hat{Z})< \epsilon
\end{equation}
\end{lemma}
\begin{proof}
Observe that this lemma follows from the general definition of mutual information \cite{inft:cover} pg. 252. Specifically, let $\mathcal{P}$ and  $\mathcal{Q}$ be finite partitions of the range of $X$ and $W$, respectively. The quantization of $X$ and $W$ by $\mathcal{P}$ and  $\mathcal{Q}$ (denoted $[X]_{\mathcal{P}}$ and $[W]_{\mathcal{Q}}$) are discrete random variables defined by 
\begin{align}
Pr([X]_{\mathcal{P}}=i)&=Pr(X \in P_i)=\int_{P_i} dF(x) \\
Pr([W]_{\mathcal{Q}}=i)&=Pr(W \in Q_i)=\int_{Q_i} dF(w)
\end{align} 
Then by definition, 
\begin{align}
I(X;W)&=\sup_{\mathcal{P},\mathcal{Q}} I([X]_{\mathcal{P}};[W]_{\mathcal{Q}}) \\
&=\sup_{\mathcal{Q}} I(X;[W]_{\mathcal{Q}}).
\end{align}
Where the second equality follows trivially since $| \mathcal{X} |=2$. Since $X \rightarrow W \rightarrow [W]_{\mathcal{Q}_L}$ is a Markov chain, , any A/D(L) converter with associated partition $\mathcal{Q}_L$ will have $I(X;[W]_{\mathcal{Q}_L}) \leq I(X;W) $. Thus given $\epsilon > 0$, by definition of $\sup$, there exists a finite partition of the range of $W$, say $\mathcal{Q}$, such that $$ I(X;W)-I(X; [W]_{\mathcal{Q}})< \epsilon. $$
Observe that $[W]_{\mathcal{Q}}$ is just $\hat{Z}$ for a particular A/D(L) converter where $L=|\mathcal{Q}|$. In short,  discretizing $W^N$ corresponds to an adequately chosen A/D(L) converter, which is precisely $\hat{Z}^N$.   

Note $\hat{Z}^N$  and $W^N$ are referenced from Fig. 2, where we use the notation $\hat{Z}^N$  when $L>2$. The dependence on $L$ is implicit in $\hat{Z}^N$.  
\end{proof}

Observe by the approximation lemma, it will suffice to approximate the equivocation associated with the A/D(L) by the equivocation given by direct access to the AWGN channel. It is simpler to calculate $I(X;W)$ instead of $I(X;\hat{Z})$. Then using the approximation lemma, we have,
\begin{corollary}\label{MEL}(Max Equivocation Loss)
\begin{equation}\label{MELe}
\lim_{L \rightarrow \infty} \frac{h(p_W)-1+I(X;\hat{Z})}{h(p_W)-h(p)}=\frac{h(p_W)-1+I(X;W)}{h(p_W)-h(p)}
\end{equation}

\end{corollary}
\begin{proof}
This follows by direct application of the approximation lemma to the equivocation loss theorem.
\end{proof} \par
Now we are left with calculating $I(X;W)$. First we shall calculate the cdf of $W$.
\begin{align}
P(W \leq w) =& P(X=-1)P(W \leq w | X =-1) \\
&+ P(X=1)P(W \leq w | X =1) \\
=& \frac{1}{2} \left[ \Phi \left(\frac{w+1}{\sqrt{\sigma_M^2+\sigma_W^2}} \right) +\Phi \left(\frac{w-1}{\sqrt{\sigma_M^2+\sigma_W^2}} \right)\right]
\end{align}
Therefore taking the derivative,
\begin{align}
f_W(w) =& \frac{d}{dw} P(W \leq w) \\
=&\frac{\phi \left(\frac{w+1}{\sqrt{\sigma_M^2+\sigma_W^2}} \right) +\phi \left(\frac{w-1}{\sqrt{\sigma_M^2+\sigma_W^2}} \right)}{2\sqrt{\sigma_M^2+\sigma_W^2}}
\end{align}
Then by the definition of mutual information,
\begin{align}
I
&(X;W)= H(W) - H(W|X) \\
=& H(W) - P(X=-1)H(W|X=-1) \\
&-P(X=1)H(W|X=1) \\
=& H(W) - \frac{1}{2} \log 2 \pi e (\sigma_M^2 +\sigma_W^2) \\
=& - \int_{- \infty}^{\infty} f_W(w) \log f_W(w) dw - \frac{1}{2} \log 2 \pi e (\sigma_M^2 +\sigma_W^2)
\end{align} \par
\begin{figure}[!htbp]
\centering
\includegraphics[width=3in]{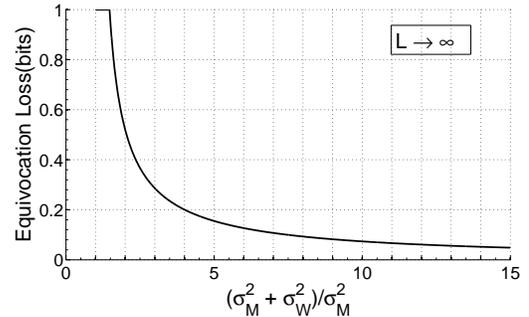}
\caption{Maximum Equivocation Loss ($\sigma_M^2=1$)}
\label{mloss}
\end{figure}
We have calculated the mutual information, and can now use corollary \ref{MEL} to estimate the loss in equivocation. Looking at  (\ref{MELe}) we see as the noise power in the wiretap channel grows, that the equivocation loss goes to zero, which is illustrated in Fig. 3. On the other hand, as the noise power in the wiretap channel gets smaller, secrecy is lost. In particular, if the $\sigma_W^2 = \sigma_M^2 $, approximately $0.5$ bits of secrecy is lost per bit transmitted. The equivocation is less than half of the value for the case of a two level A/D converter. In this particular instance, where the eavesdropper is allowed to use arbitrary precision analog to digital converters, it is difficult to guarantee a particular amount of secrecy. \par
 
The analysis presented here is for a simple model, and is present only to illustrate the difficulty in making assumptions about the physical limitations in channel access of an eavesdropper. In the analysis of more complicated scenarios, especially in the wireless domain, incorrect assumptions on the eavesdropper will present similar security losses. On the other hand, this does not constitute a proof that physical layer security using Wyner's method is not possible. \par
  
\section{ USING SHARED KEYS TO INCREASE SECURITY}
In the previous section, we illustrated what happens if the intended users believe that the eavesdropper uses an A/D(2) converter, but in reality the eavesdropper uses an A/D(L) converter. While this is, in many regards, a toy example, it illustrates that a wrong assumption (belief) on the eavesdropper's channel leads to a compromised security system. In other words, a system is information-theoretically secure only if the assumptions (made by the intended users) of the eavesdropper's channel are correct. However, in practice, the intended users do not control the eavesdropper's channel nor behavior. Hence the intended users can never verify their assumptions.  \par
In a practical setting, apart from using a better A/D converter, the eavesdropper can; 1) use multiple antennas, 2) decrease the distance to the transmitter, 3) use non-AWGN channels, etc... If the intended users make the wrong guess regarding any of these parameters, they will not achieve a secure system.  \par   
Thus making assumptions about the wiretapper's capabilities is difficult. One potential alternative is to combine some of the techniques of modern cryptography with Wyner's secrecy coding techniques. We will artificially create random channels, and then use the corresponding stochastic encoding to secure the message. Our assumption will be, as in contemporary cryptography, that recovering the preshared key used to create the channel is computationally intractable.\par
We remove the dependence on a true source of randomness. In particular, the stochastic encoder is dependent on a source of true randomness. Our solution will be to use general purpose pseudorandom number generators as input to the stochastic encoder and input to create artificial channels. By pseudorandom number generation, we follow the guideline of Goldreich\cite{rndm:gold}, \\
\textit{"Loosely speaking, general-purpose pseudorandom generators are efficient deterministic programs that expand short randomly selected seeds into longer pseudorandom bit sequences, where the latter are defined as computationally indistinguishable from truly random sequences by any efficient algorithm."}  \par  

\subsection{General Solution Strategy}

Our general solution strategy is to simulate the main and wiretap channels in the Wyner wiretap channel model at the transmitter prior to transmission, see Fig. \ref{psol}. The resulting data is then transmitted over a physical channel with channel coding such that it is assumed the eavesdropper and intended users have perfect copies of the sent data. When we simulate the wiretap channel, it is done in such a way that the intended receiver with knowledge of a preshared key can reverse the simulation. In addition, at no point do we assume a true source of randomness.\par

The first step in designing our system, is to choose the preferred channel types for the main and wiretap channel, where the wiretap channel is stochastically degraded with respect to the main channel. We choose an encoding scheme that will achieve information theoretic security based on the assumption that the real channel model was implemented.  Up to this point, we are exactly in the framework of Wyner's paper. This is where we take a separate path. We use pseudorandom number generation as the source of randomness for the stochastic encoding, and to artificially create the main channel. The pseudorandom noise generators are initialized with a randomly selected seed that only the sender knows. In order to emulate the wiretap channel, we use a general purpose pseudorandom number generator that is seeded with a preshared key. For the wiretap channel we require that given the preshared key, the intended receiver must be able to reverse the process of emulating the wiretap channel. \par
In this model the intended receiver is knowledgeable of the shared key for the wiretap noise, and can remove this noise exactly. The pseudorandom noise used to emulate the main channel, can not be removed, and the intended user decodes using the error correcting code.  At this point the intended user has recovered the source message. \par
On the other hand, the eavesdropper is knowledgable of the complete design of the system, with the exception of the preshared key. The eavesdropper's attacks will then be based on exploiting the weakness of the pseudorandom number generation to recover the secret key, and seeds.

\begin{figure}[!htbp]
\centering
\includegraphics[width=3.5in]{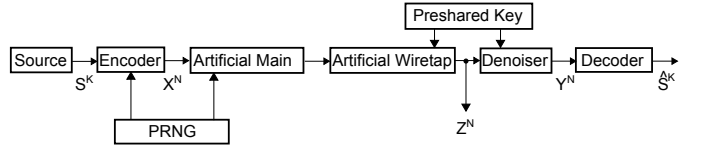}
\caption{General Model for Wyner based coding using computational intractability}
\label{psol}
\end{figure}

\subsection{Example}

We provide an example that follows the approach of Mihaljevic \cite{wtap:mishalpn} in system design. In particular, we assume the following notation:
\begin{align*}
 \bold{a} &= [a_i]_{i=1}^l : \text{plaintext} \\
  \bold{r}&=[r_i]_{i=1}^{m-l} : \text{Bernoulli}(1/2) \text{i.i.d. random variables} \\
 \bold{u} &= [u_i]_{i=1}^k  : \text{public Bernoulli}(1/2) \text{i.i.d. random variables} \\
 \bold{S} &= [s_{i,j}]_{i=1 \: j=1}^{k \:\quad n}  :\text{ binary} (k \times n) \text{secret key matrix} \\
 \bold{v} &= [v_i]_{i=1}^n  :\text{ unknown  Bernoulli}(p) \text{i.i.d. random variables } \\
 \bold{M} &: \text{invertible mixing matrix }(m \times m) \\ 
 f_E  &: \text{stochastic encoder from} \{ 0,1 \}^m \rightarrow \{0,1\}^n \\
 g  &: \text{decoder} \\
\end{align*}

Then encryption and decryption are as follows.
The ciphertext, $\bold{z}$, is calculated as
\begin{equation}
\bold{z} = f_E( \bold{M} \cdot  (\bold{a} \parallel \bold{r}) )  \oplus \bold{u} \cdot \bold{S} \oplus \bold{v} 
\end{equation}
where $\parallel$ denotes concatenation. Let $trunc( \bold{x}, l)$, be the truncation function which returns the first $l$ bits of $\bold{x}$. Decryption of the ciphertext, is shown below,
\begin{equation}
\bold{a}  = \text{trunc} \left( \bold{M}^{-1} \cdot g(\bold{z} \oplus \bold{u} \cdot \bold{S} ) , l \right)
\end{equation}

 There are three main parts in the design of this system. In particular, we first use concatenation of a random vector and a mixing matrix to increase the entropy of the input to the encoding. This helps protect against known plaintext attacks. Then we create a main channel that is a $BSC(p)$, and we employ a stochastic encoding to correct for the corresponding errors. Finally we use a secret key matrix, and public random vector to create a "noisier" channel for the eavesdropper.   The new system is in Fig. \ref{prng}. From an information theoretic point of view, since the error correction will compensate for the added randomness, this system is not secure, as the eavesdroppers channel is deterministic. From here we take the viewpoint of conventional cryptography, that solving for the secret key is computationally intractable. In particular, this is an instance of the learning parity with noise (LPN) problem, which is provably NP-hard. Furthermore the LPN problem is a natural choice for stochastic encoding techniques. This follows since solving for the secret key, $S$, is equivalent to decoding a random linear block code \cite{lpn:bk}. \par
 \begin{figure}[!htbp]
\centering
\includegraphics[width=3in]{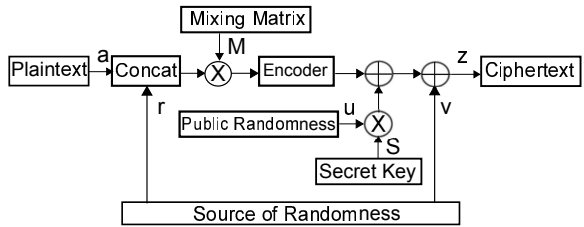}
\caption{Learning Parity with Noise Based Cryptoystem  }
\label{prng}
\end{figure}
 
Thus we are combining a computationally hard problem with an information theoretic based security method. We have now made a trade off on assumptions about the eavesdropper's channel, to one on the ability to solve a computationally hard problem. The latter is widely accepted in conventional cryptography. \par
\subsection{Implementation Issues}
A motivation for this paper, was to create a lightweight cryptographic system based on concepts from physical layer security. Two potential issues that will increase the complexity of using this system are; 1) the creation of satisfactory and cryptographically secure random variables, and 2) proper choice of coding parameters. \par
Creation of cryptographically secure pseudo-random numbers is a difficult problem, and in this security system they are directly applied to plaintext. This leaves an opening for attacks on the pseudo-random number generator (PRNG).  Furthermore a source of public randomness is assumed. This leaves the system open to man in the middle attacks. Thus we must include measures to guarantee availability and authenticity of the public random vector. \par
In addition, we must choose adequate coding parameters for the length of the $\bold{r}$ vector, and the stochastic encoder, $f_E$. A thorough analysis should be performed to find a characterization of the eavesdropper's channel with regards to Wyner's wiretap model. In particular, assuming the randomness in $\bold{v}$ is cryptographically secure, we say the main channel is a $BSC(p)$ channel. Then we may find a $0<p_W< \frac{1}{2}$ as large as possible and base the construction of $f_E$ on the assumption that the eavesdropper's channel is a $BSC(p_W)$. Observe choosing $p_W$ large, increases the rate of our stochastic encoder, which reduces the security of our system.  Thus we must find an upper limit for $p_W$ such that the resulting parameters for $\bold{r}$ and $f_E$ ensure it is computationally infeasible for the eavesdropper to recover the secret key. \par

\section{Conclusion}
In this paper, one of the concerns in basing a security system on assumptions about the eavesdropping capabilities of an unintended user has been adressed. An alternative has been presented which uses one of the main principles of Wyner's paper, security through stochastic encoding. It is a convenient solution, in the sense that it does not require any changes in infrastructure. In particular, it can be added on at the application layer. It is noted, that this is a toy example, and used to illustrate a point. \par

A valid counter argument to the solution proposed in this paper, is that there is still an assumption made on the eavesdropper's capabilities in using this method. The onus has been transferred from relying on physically acquiring the signal to computational limitations. In regard to this argument, this paper than serves two purpose. First, it shows that care must be taken in assumptions about the eavesdropper. In addition, for a simple case it illustrates how to calculate the effect of incorrect assumptions, and engineer a system within certain secrecy tolerances.

\section*{Acknowledgment}
The author would mainly like to thank Prof. Aleksandar Kavcic for his patience, and assistance in writing this technical paper.

\ifCLASSOPTIONcaptionsoff
  \newpage
\fi

\end{document}